\documentclass[11pt, letterpaper]{article}
\usepackage[margin=1in]{geometry}

\usepackage{amsfonts, amsmath, amssymb, amsthm}
\usepackage{bm}
\usepackage{verbatim}
\usepackage{color}
\usepackage{euscript}
\usepackage{graphicx}
\usepackage[usenames,dvipsnames]{xcolor}
\usepackage{paralist}
\usepackage{subcaption}
\usepackage{xspace}
\usepackage{wrapfig}
\usepackage{url}
\usepackage{mathtools}

\newtheorem{theorem}{Theorem}
\newtheorem{lemma}{Lemma}

\newcommand{\vare}{\varepsilon}

\newcommand{\cQ}{\mathcal{Q}}
\newcommand{\cC}{\mathcal{C}}
\newcommand{\cX}{\mathcal{X}}
\newcommand{\cY}{\mathcal{Y}}
\newcommand{\cZ}{\mathcal{Z}}

\newcommand{\sort}{O\left(\frac{N}{B}\log_{M/B}\left(\frac{N}{B}\right) \right)}
\newcommand{\bound}{\mathcal B}
\newcommand{\List}{\mathcal L}
\newcommand{\expire}{\mathit{exp}}

\newcommand{\tcr}[1]{\textcolor{red}{#1}}

\newcommand{\commented}[1]{}

\begin{document}

\title{Optimal-Cost Construction of Shallow Cuttings for 3-D Dominance Ranges in the I/O-Model} 

\date{}

\author{Yakov Nekrich\thanks{Department of Computer Science, Michigan Technological University. Supported by the National Science Foundation under NSF grant 2203278}
\and Saladi Rahul\thanks{Department of Computer Science and Automation, Indian Institute of Science. Supported in part by the Walmart Center for Tech Excellence at IISc (CSR Grant WMGT-230001) and 
Anusandhan National Research Foundation (SP/ANRF-24-0043 CRG/2023/005776).}}

%\keywords{Data Structures, I/O-efficient algorithms, Orthogonal Range Searching.} 

\maketitle

\begin{abstract}
Shallow cuttings are a fundamental tool 
in computational geometry and spatial databases
for solving offline and online range searching problems. 
For a set $P$ of $N$ points in 3-D, at SODA'14, Afshani and Tsakalidis designed an optimal $O(N\log_2N)$ time algorithm that  constructs shallow cuttings for 3-D dominance ranges in internal memory. Even though shallow cuttings are used in the I/O-model to design space and query efficient range searching data structures, an efficient construction of them is not known till now. In this paper, we design an optimal-cost algorithm to construct shallow cuttings for 3-D dominance ranges. The number of I/Os performed by the algorithm is $\sort$, where 
$B$ is the block size and $M$ is the memory size. 

As two applications of the optimal-cost construction algorithm, 
we design fast algorithms for offline 3-D dominance reporting
and offline 3-D  approximate dominance counting. 
We believe that our algorithm will  find further applications in offline 3-D range searching problems and 
in improving construction cost of data structures for 3-D
range searching problems.
\end{abstract}

%\newpage

%=================
%Introduction

\section{Introduction}

{\em Shallow cuttings} are one of the most fundamental tools used
for designing range searching data structures
in computational geometry~\cite{p08b,aal09,act14,ahz10,
nr23} and spatial databases~\cite{rt15,rt16,aaefv98}. 
The main contribution of this paper is the 
first known optimal-cost construction of shallow cuttings
for {\em 3-D dominance ranges} in the {\em I/O-model}. As 
a consequence, we obtain efficient algorithms for the {\em offline 
3-D dominance reporting} problem and the {\em offline
3-D approximate counting} problem.

\vspace{0.1 in}
\noindent\textbf{Shallow cuttings for 3-D dominance ranges.}
%\paragraph{Shallow cuttings for 3-D dominance ranges.}
Consider two 3-D points $p$ and $q$. 
A point $p$ {\em dominates} another point $q$ if and only if 
$p$ has a higher coordinate value than $q$ in each dimension.
Let $P$ be a set of $N$ points in 3-D.
A {\em $k$-level shallow cutting} of $P$~\cite{p08b} 
(See Figure~\ref{fig:3-D-sc}(a))
is a collection $\cC$ of {\em cells} of the form $(-\infty, x] \times (-\infty, y] \times (-\infty,z]$
 satisfying the following three properties:
\begin{enumerate}
\item The number of cells is  $O(n/k)$, i.e., $|\cC|=O(n/k)$.
\item Each cell in  $\cC$ contains at most $10 k$ points from $P$.
\item Any 3-D point  that  dominates at most $k$ points of $P$ 
lies inside at least one cell in $\cC$.
\end{enumerate}
\begin{figure}[h]
  \centering
  \includegraphics[scale=1.2]{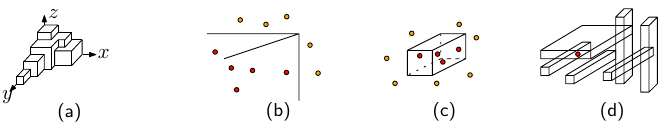}
  \caption{(a) A schematic representation of a $k$-level 
  3-D shallow cutting, (b) 3-D dominance reporting and 
  approximate counting query,
  (c) 3-D orthogonal range reporting query, (d) 3-D 
  rectangle stabbing query.}
  \label{fig:3-D-sc}
\end{figure}

\vspace{0.1 in}
\noindent\textbf{I/O-model of computation.} 
%\paragraph{I/O-model of computation.} 
We will construct shallow cuttings
for 3-D dominance ranges in the I/O-model of computation~\cite{v06}.
In the era of big-data, it is natural that the input will not 
fit completely in the internal memory, and hence, resides
in the hard-disk (or the external memory). In this model, 
the internal memory has $M$ words, and the hard-disk has been
formatted into {\em blocks} of $B$ words.
An I/O operation exchanges a block of data between the disk and the memory. 
The space of a data structure is the number of blocks occupied
in the external memory, 
and the running time of an algorithm is the number of I/Os 
performed. The CPU calculation is free. 
A natural assumption is that $M=\Omega(B)$.

\vspace{0.1 in}
\noindent\textbf{3-D orthogonal range searching.} 
%\paragraph{3-D orthogonal range searching.} 
Orthogonal range
searching and related problems are foundational problems on which the
fields of computational geometry and spatial databases
grew. Naturally, they have been 
well studied in the I/O-model of computation. 
Consider the following problems for which the state-of-the-art data structures in the I/O-model 
extensively rely on the clever use of 
shallow cuttings for 3-D dominance ranges:

\begin{enumerate}
    \item In the 
    {\em 3-D orthogonal range reporting} problem 
    (Figure~\ref{fig:3-D-sc}(c)), the input is a 
    set $P$ of $N$ points in 3-D. Preprocess $P$ into a space-efficient
    data structure, so that given an axis-aligned query box 
    $q=[x_1,x_2] \times [y_1,y_2] \times [z_1,z_2]$, the goal 
    is to quickly report $P\cap q$ (the points of $P$ lying inside
    $q$). AT FOCS'09, Afshani, Arge, and Larsen~\cite{aal09}
    presented a data structure for this problem with space 
    $O\left(\frac{N}{B}\left(\frac{\log N}{\log\log_BN}\right)^3\right)$
    and query I/Os $O(\log_B N + k/B)$, where $k$ is the 
    number of points reported.
    \item The building block 
    for the above problem is the {\em 3-D dominance 
    reporting} problem (Figure~\ref{fig:3-D-sc}(b)), where 
    the query is a 3-D dominance range $q=(-\infty,x] \times (-\infty,y] 
    \times (-\infty,z]$. At ESA'08, Afshani~\cite{p08} presented
    a data structure for this problem with space $O(N/B)$ and 
    query I/Os $O(\log_B N + k/B)$. This is optimal in terms 
    of space and query I/Os.
    \item At SoCG'09, Afshani, Hamilton, and Zeh presented a 
    general approach to design data structures to answer
    the {\em relative $(1+\vare)$-approximate counting} query
    for a general class of problems including the 3-D dominance
    setting. Here the goal is to output the 
    estimate of $|P \cap q|$. 
    An optimal solution in terms of space 
    and query I/Os was obtained.
    \item {\em Rectangle stabbing reporting} 
    (Figure~\ref{fig:3-D-sc}(d))
    is the ``reverse'' of orthogonal range reporting where the 
    input is a set of axis-aligned 3-D boxes and the query 
    is a 3-D point. Arguably, rectangle stabbing reporting 
    is as important as orthogonal range reporting and is well
    studied in the literature. The current state-of-the-art results 
    in internal memory models of computation for rectangle 
    stabbing reporting by Rahul~\cite{r15} at SODA'15 and 
    Chan, Nekrich, Rahul, and Tsakalidis~\cite{cnrt18} at 
    ICALP'18 rely on shallow cuttings for 3-D dominance 
    ranges. It is unlikely that a linear-space 
    data structure in the I/O-model can be constructed 
    without using shallow cuttings.
\end{enumerate}

\vspace{0.1 in}
\noindent\textbf{Lack of fast construction of data structures.}
%\paragraph{Lack of fast construction of data structures.} 
An important criteria in the design of the data structures 
is to optimize the I/Os required to construct the data structure.
A severe limitation of {\em all} the above solutions is the 
lack of such fast construction algorithms. In fact, quoting
from one of the above papers (Afshani, Hamilton, and Zeh~\cite{ahz10}),

\vspace{0.1 in}

{\em ``All data structures presented in this paper are static, and no efficient construction methods for these structures are known in the I/O or 
cache-oblivious model. The main obstacle is the lack of an I/O-efficient or 
cache-oblivious construction procedure for shallow cuttings.''}

\vspace{0.1 in}

We believe that our work on optimal-cost construction of 
shallow cuttings for 3-D dominance ranges 
will trigger further work to design 
fast construction of data structures for the above mentioned
problems.

\vspace{0.1 in}
\noindent\textbf{Lack of optimal solutions for 
offline 3-D orthogonal range searching.}
%\paragraph{Lack of optimal solutions for 
%offline 3-D orthogonal range searching.} 
In the {\em offline} orthogonal range searching problem, 
we are given a set $Q$ of queries upfront 
along with the point set $P$.
Many 2-D offline problems in computational geometry (including 
range searching type problems) have been optimally 
solved in the I/O-model. See the survey paper by Vitter 
(Theorem~8.1 in \cite{v06} provides a comprehensive list of 
more than ten problems in 2-D which are solved I/O-optimally).

However, to the best of our knowledge, offline 3-D orthogonal range 
searching problems have not received
the same attention in spite of their fundamental nature. 
Again, the lack of optimal-cost shallow cuttings in 3-D
is the major bottleneck. 

%\vspace{0.1 in}
%\noindent\textbf{Our results.} 
\subsection{Our results}
The main result in this paper 
is an optimal-cost construction of the $k$-level 
shallow cutting for 3-D dominance ranges in the I/O-model.

\begin{theorem}
  \label{thm:shallow1}
The $k$-level shallow cutting for 3-D dominance ranges 
on $N$ points can be 
constructed in $\sort$ I/Os. The construction cost is optimal
in the I/O-model.
\end{theorem}

\begin{comment}
As a bonus, we also construct an ``implicit'' representation 
of several levels in the shallow cutting. We believe
that this result can find applications in future work.
\tcr{is implicit defined here?}

\begin{theorem}
  \label{thm:shallow2}
    Let $\rho$ denote an arbitrary positive constant and let $f_i=\rho^i$ for $i=1$,$\ldots$, $\log_{\rho}N$.  We can construct implicit representations of all $f_i$-shallow cuttings in $O\left(sort(N) +\frac{N}{B}\log\log \frac{M}{B}\right)$ I/Os.  
\end{theorem}
\end{comment}

To demonstrate the impact of our new construction, 
we use shallow cuttings for 3-D dominance ranges to 
obtain efficient algorithms for
two fundamental problems: {\em offline 3-D dominance approximate 
counting} and {\em offline 3-D dominance reporting}.

\begin{theorem}\label{thm:main-approx-cnt}
\textbf{(Application-1)} Let $P$ be a set of $N$ points in 3-D and 
$Q$ be a set of 3-D dominance ranges. In the
offline 3-D dominance approximate counting
problem, for any  $q\in Q$, the goal is to return a value 
in the range $[(1-\vare)|P \cap q|,|P\cap q|]$, where 
$\vare \in (0, 1)$ is a fixed value.
There is a deterministic algorithm to solve the problem
in $O_{\vare}\left(\frac{|P|}{B}\left(\log_{M/B}\frac{|P|}
{B}\right)^{O(1)} + \frac{|Q|}{B}\log_{M/B}\frac{|P|}{B} \right)$
I/Os. The notation $O_{\vare}(\cdot)$ hides the dependency on 
$\vare$.
\end{theorem}

The key feature in the above theorem 
is that the construction cost bound has
the term $M/B$ in the base of the logarithm which is desirable
in the I/O-model (instead of {\em base-$2$} obtained from straightforward
adaptation of internal memory algorithms
or {\em base-$B$} obtained from straightforward adaption of 
external memory data structures for online problems).
In terms of techniques, even though the final output is a 
relative approximation, the crucial ingredients in the algorithm
are an additive error counting structure that is combined with 
a non-trivial {\em van-Emde-Boas-style recursion}.

\begin{theorem}\label{thm:main-reporting}
\textbf{(Application-2)} Let $P$ be a set of $N$ points in 3-D and 
$Q$ be a set of 3-D dominance ranges. 
In the offline 3-D dominance reporting 
problem, for any $q\in Q$, the goal is to report 
$P \cap q$. The problem can be solved in $O\left(\frac{|P| + |Q|}{B}\log_{M/B}\frac{|P|}{B} + \frac{K}{B}\right)$ I/Os, 
where $K$ is the output size.
\end{theorem}
The final algorithm is a combination of several ideas:
filtering search (to charge the I/Os performed for a query 
to  its output-size), 
optimal $k$-level shallow cuttings (which were not 
known before), roughly $O(\log\log N)$ levels of 
shallow cuttings, and a cleverly chosen base-case.
The non-trivial part of the analysis is bounding
the number of I/Os to perform the so-called offline-find-any 
procedures (defined later). 
%Due to lack of space,
%the proof of Theorem~\ref{thm:main-reporting} has
%been moved to Appendix~\ref{sec:reporting}.

To prove the above theorems, we make use of data structures
with specific requirements. For example, in Theorem~\ref{thm:reporting} data structures with fast 
construction I/Os are considered and in Theorem~\ref{thm:add-approx}
a data structure with large additive error is considered.
Although there exist data structures with desired 
bounds to handle 
some of these queries~\cite{gi98,ks99,r81,t12b}, we are not aware of any prior work that obtains these results
using a common framework, which might be useful for other problems.

%%%%Key results of supporting structures %%%%%%%
\begin{theorem}\label{thm:reporting}
\textbf{(Reporting, counting, and selection)} Let $P$ be a set of $N$ points in 3-D
and $\gamma$ be a fixed constant in $(0,1)$. Then 
there exists a data structure that can be constructed
in $O(sort(N))=\sort$ I/Os and then used to answer a
\begin{itemize}
\item 3-D dominance reporting query in $O((N/B)^{\gamma} + t/B)$
I/Os, where $t$ is the number of points reported.
Given a  3-D dominance query range $q$, the goal is to report 
the points in $P\cap q$.
\item 3-D dominance counting query in $O((N/B)^{\gamma})$ I/Os. 
Given a  3-D dominance query range $q$, the goal is to report 
$|P\cap q|$.
\item 3-D dominance $x,y,z$-selection query in $O((N/B)^{\gamma})$ I/Os. In an $x$-selection query, given a query point $q=(q_x,q_y,q_z)$
and an integer $k'\in [1,N]$,
the goal is to return a point $q_1=(q'_x,q_y,q_z)$ that
dominates $k'$ points in $P$. Analogously, in a 
$y$-selection query (resp., $z$-selection query), 
the goal is to return a point $q_2=(q_x,q'_y,q_z)$
(resp., $q_3=(q_x,q_y,q'_z)$) that dominates $k'$ points in $P$.
\end{itemize}
\end{theorem}

\begin{theorem}\label{thm:add-approx}
\textbf{(Offline additive error counting)} 
Let $P$ be a set of $N$ points 
in 3-D and $Q$ be a set of 
3-D dominance ranges.
 In the offline additive error counting problem, for each query $q\in Q$, the goal is to return a value 
$n_q$ such that $|P\cap q|-n_q \leq \vare N^{2/3}B^{1/3}$.
The data structure can be constructed using 
$O_{\vare}(sort(N))$ I/Os and the queries are answered 
by using
$O_{\vare}\left(sort(N) + sort(|Q|)\right)$ I/Os. 
\end{theorem}

Our solution for additive error 
query can be extended to handle smaller values of error
(such as $\vare\sqrt{NB}$). However, for our purpose
an additive error of $\vare N^{2/3}B^{1/3}$ suffices.

The remaining paper is structured as follows. In Section~\ref{sec:prelim}, 
we discuss some preliminaries. In Section~\ref{sec:prior-const}, we 
give an overview of two previous constructions of shallow cuttings for 
3-D dominance ranges in internal memory. Then, in Section~\ref{sec:our-algo} 
we present our optimal-cost construction in the I/O-model. This is followed
by two applications of our construction, offline 3-D dominance approximate counting 
(in Section~\ref{sec:approx}) and offline 3-D dominance reporting (in Section~\ref{sec:reporting}). In Section~\ref{sec:support}, we present the details of our supporting structures (Theorem~\ref{thm:reporting}
and Theorem~\ref{thm:add-approx}). Finally, potential directions for future work are discussed in Section~\ref{app:future}.

\section{Preliminaries}\label{sec:prelim}
\noindent\textbf{Definitions.}
%\paragraph{Definitions.} 
The {\em conflict list} of a cell $C \in \cC$ is defined as the points of 
$P$ that are inside $C$. The {\em apex point} of a cell or a 3-D dominance range
$(-\infty, x] \times (-\infty, y] \times (-\infty,z]$ is defined as 
$(x,y,z)$. The {\em depth} of a point $p=(x_p, y_p, z_p)$ w.r.t. 
a pointset $P$ is the number of points of $P$ lying in the cell 
$(-\infty, x_p] \times (-\infty, y_p] \times (-\infty,z_p]$.
%\tcb{A point $p=(x_p, y_p, z_p)$ 3-D dominates another point $q=(x_q,y_q,z_q)$
%if and only if $x_p > x_q, y_p > y_q \text{ and } z_p > z_q$.}

\begin{figure}[h]
  \centering
  \includegraphics[scale=1]{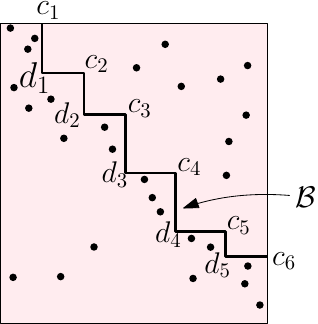}

  \caption{\label{fig:2-D-sc} An example of a $k$-level 2-D shallow cutting. 
  Here $k=3$, so each inner corner $d_i$ dominates three points and 
  each outer corner $c_i$ dominates six points. }
\end{figure}

\vspace{0.1 in }
\noindent\textbf{Shallow cuttings for 2-D dominance ranges.}
%\paragraph{Shallow cuttings for 2-D dominance ranges.} 
As a warm-up, we first introduce shallow cuttings for 
{\em 2-D dominance ranges}. Let $P$ be a set 
of $N$ points in 2-D. The collection $\cC$ 
of cells in the $k$-level shallow cutting of $P$  
will be of form $(-\infty,x] \times (-\infty, y]$.
Consider the union of the cells in $\cC$. 
We will focus on constructing the {\em outer boundary} 
of the union that is shown in Figure~\ref{fig:2-D-sc}. 
The outer boundary $\bound$ is a $1$-dimensional, 
monotone, and orthogonal curve. Specifically, 
$\bound$ starts from
$y=+\infty$ and consists of alternating
vertical and horizontal segments. Let the sequence of endpoints 
of $\bound$ be
$c_1,d_1,c_2,d_2,\ldots, d_{t-1} \text{ and } c_t$. 
The $c_i$'s are the {\em outer corners} and the $d_i$'s are
the {\em inner corners}. The number of points of $P$ 
dominated by any outer corner (resp., inner corner) will 
be at most $2k$ (resp., equal to $k$). 
The $k$-level shallow cutting for 2-D dominance ranges
can be constructed in $O(sort(N))$ I/Os.
(Note that the $2k$ bound is tighter than the $10k$ bound 
we use for 3-D shallow cuttings.)

\vspace{0.1 in}
\noindent\textbf{Offline-find-any procedure.}
%\paragraph{Offline-find-any procedure.} 
The input to the {\em offline-find-any} problem is a collection of 
cells $\cC$ in the $k$-level shallow cutting of $P$,
for any $1\leq k\leq n$, and a set $Q$ of 3-D query points.
In the offline-find-any procedure, 
for each $q\in Q$, the goal is to find {\em any} cell in $\cC$
that contains $q$ or report that none of them contain $q$.
The distribution-sweeping approach in the survey paper of 
Vitter~\cite{v06} can be used to solve this problem optimally.

\begin{lemma}\label{lem:find-any}
Given a set $\cC$ of cells in 3-D and a set $Q$ of 
query points in 3-D, the offline-find-any procedure
requires 
$O\left( \left( \frac{|\cC| + |Q|}{B}\right)\log_{M/B}\frac{|\cC|}{B}\right)$ I/Os.
\end{lemma}

\section{Overview of previous construction algorithms
%of shallow cuttings  for dominance ranges
}
\label{sec:prior-const}
In this section we give an overview of a couple
of previous constructions whose ideas will be required
in our final solution.
%\noindent\textbf{Afshani and Tsakalidis's construction (AT-construction).}
\subsection{Afshani and Tsakalidis's construction (AT-construction)}
The problem becomes challenging in 3-D. 
At SODA'14, Afshani and Tsakalidis~\cite{at14} obtained an 
optimal internal memory algorithm that constructs a 
$k$-level shallow cutting in
$O(N\log_2 N)$ time. Their algorithm constructs 
the outer boundary $\bound$ of the union of the cells in 
$\cC$. In 3-D the outer boundary is a $2$-dimensional, 
monotone, and orthogonal surface (see Figure~\ref{fig:3-D-sc}(a)). 
The algorithm combines the  sweep-plane 
approach with a method of maintaining a 2-D $k$-level 
shallow cutting under deletions. 
Let $z_1,\ldots,z_N$ be the sorted sequence of $P$ 
in decreasing order of their $z$-coordinate values.
Starting from $z=z_1$, a plane parallel to $xy$-plane
is moved in the $-z$ 
direction. The {\em invariant} in the algorithm is to
maintain a 2-D shallow cutting for the 
$xy$-projections of points of $P$ below the sweep-plane,
by ensuring that  the depth of each (inner or outer) corner 
in the 2-D shallow cutting 
lies in the range $[k,10k]$.
When the sweep-plane reaches $z_i$, then $z_i$ is deleted 
and if the invariant is violated, then the 2-D shallow cutting of $xy$-projections 
of $z_{i+1}, \ldots, z_n$ is updated 
via a procedure called {\em patching}. See Figure~\ref{fig:patching}.
The precise details 
of the patching procedure are not necessary for our final 
algorithm. 
\begin{figure}
  \centering
  \includegraphics[scale=0.7]{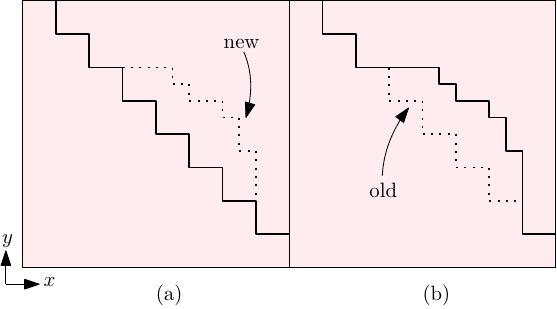}
  \caption{\label{fig:patching} (a) Outer boundary 
  (in bold) before the patching procedure. The dotted 
  boundary is the new curve created during the 
  patching procedure. (b) Outer boundary (in bold) after the 
  patching procedure. The dotted curve was
  part of the outer boundary before patching. Each outer 
  corner created during the patching procedure corresponds
  to the apex point of a 3-D cell in the $k$-level cutting of $P$.}
\end{figure}

The running time of their algorithm can be summarized by 
the following equation:
\begin{equation}
T_1(N) = \underbrace{T_{sort}(N)}_{\text{sorting}} + \underbrace{\Theta(N)\cdot T_{pred}(N/k)}_{\text{bottleneck 1}} + \underbrace{O(N/k)\cdot T_{aux}(N)}_{\text{bottleneck 2}}
\end{equation}

There are two bottlenecks in adapting their algorithm in the 
I/O-model. At any point in the sweep-plane, 
let $\List$ be the ordering of the corners in the 
2-D shallow cutting
based on their $x$-coordinate values. 
Each time a point $z_i \in P$
is deleted, a predecessor search is performed on $\List$. Also,
the patching procedure leads to insertion and deletion of corners
into $\List$.  The predecessor search is performed 
to identify all the corners in $\List$ dominating $z_i$ and 
then to update the conflict lists of those corners 
(at no extra asymptotic cost).
This is the first bottleneck, since it is not 
possible to perform $\Theta(N)$ updates and $\Theta(N)$ predecessor queries 
to the list $\List$ in $sort(N)$ I/Os. (In fact, in \cite{bfgmmt14} it was established that 
the easier problem of batched predecessor queries requires
$T_{pred}(N)=\Omega(\log_BN)$ I/Os per query when the preprocessing cost  is bounded 
by a polynomial in $N$.)
It is not clear how to reduce the number of predecessor queries
to $o(N)$, and the lower-bound discards
the standard strategy in the I/O-model 
of performing predecessor search queries in batches.

The {\em second bottleneck} is the 
lack of I/O-efficient 
data structures for performing the patching procedure.
The patching procedure involves constructing data structures
for  3-D dominance $x$-selection and  dominance $y$-selection
queries on $P$, and querying them $O(N/k)$ times. 
For our purpose, the data structures have to be constructed 
in $O(sort(N))$ I/Os. In Section~\ref{sec:support}, we present
such data structures with 
 query I/Os $T_{aux}(N)=O((N/B)^{\gamma})$ for some $\gamma \in (0,1)$ 
 (Theorem~\ref{thm:reporting}). Therefore, the patching will require
$O\left(\frac{N}{k}\left(\frac{N}{B} \right)^{\gamma}\right)$ I/Os, 
which is  expensive. 

%\vspace{0.1 in}
%\noindent\textbf{Nekrich and Rahul's construction (NR-construction).}
\subsection{Nekrich and Rahul's construction (NR-construction)}
At SODA'23, Nekrich and Rahul~\cite{nr23} designed a data structure for 4-D 
dominance reporting query by
constructing several shallow cuttings for 3-D dominance 
ranges. However, for their application, the conflict lists 
of the cells were not needed, and only the cells needed
to be constructed. They modified the
AT-construction to design an algorithm 
whose running time is proportional to the number of cells
constructed (specifically 
$O(N\log_2N)$ cells were constructed on different subsets of $P$ 
in $O(N\log_2^5N)$ time). 

For a cell $C$ in the shallow cutting with apex point 
$(x,y,z)$, let $\expire(C)$ be the largest $z$-coordinate 
such that $(x,y,\expire(C))$ dominates less than $k$ points of $P$.
Then $\expire(C)$ is defined as the {\em expiry} of $C$. 
 For a cell $C$, its expiry $\expire(C)$ can be computed via a 3-D dominance
 $z$-selection query on $P$. 
Similar to AT-construction approach, starting from $z=+\infty$, the algorithm sweeps a plane parallel to the $xy$-plane in the $-z$ direction 
and the same invariant is maintained. 
The key difference is in the {\em event points} visited by 
the sweep-plane algorithm. In the NR-construction, 
the  event points
will be the $\expire(C)$ values of all cells $C$ in the shallow cutting,
which are stored in a simple priority-queue $\cQ$. 

At each time step,
we find the event  with the largest $z$-coordinate $z_{\max}$ in $\cQ$
and set the $z$-coordinate of the sweep-plane to $z_{\max}$. When the 
cell $C$ corresponding to $z_{max}$  expires, 
(inner and outer) corners corresponding to $C$ have to be removed from 
the 2-D shallow cutting at $z_{max}$. Therefore, the patching procedure is performed. After finishing 
the patching procedure, (a) the events corresponding to the new inner
corners in the 2-D shallow cutting are inserted into $\cQ$ and the events 
corresponding to the deleted corners in the 
 2-D shallow cutting are removed from $\cQ$, (b) for each newly created 
 outer corner $(x,y)$ in the patching procedure, a new cell with apex point
 $(x,y,z_{max})$ gets {\em created}. We define $cre(C)=z_{max}$ for 
 each cell $C$ created.
 The algorithm  finishes when  the $z$-coordinate of the sweep-plane reaches
  $z=-\infty$.
We omit the details of the NR-construction
algorithm that are not needed to obtain  our result. 
The following lemma summarizes the properties of the NR-construction that 
are relevant for our algorithm.
%The precise details of NR-construction are not essential. 
%The following lemma summarizes what is needed  for our final algorithm. 

\begin{lemma}\label{lem:summary}
Nekrich and Rahul's construction~\cite{nr23} 
 of $k$-level shallow cuttings for 3-D dominance ranges 
 can be summarized as follows:
 \begin{enumerate}
 \item  It involves 
construction of a priority queue and 3-D dominance $x,y,z$-selection
data structures.\footnote{Nekrich and Rahul~\cite{nr23} originally
used {\em 3-D dominance range counting} data structures as well,
where the goal is to report $|P \cap q|$ for
a 3-D dominance range $q$. However, both queries are almost equivalent: we can answer a 3-D dominance range selection query using $O(\log N)$ range counting queries (binary search). Conversely, we can also answer 
a 3-D dominance range counting query  via $O(\log N)$  dominance $z$-selection queries.}
\item 
The number of updates and queries on the priority-queue 
is $O(N/k)$ and the number of 
3-D dominance $x,y,z$-selection queries is $O(N/k)$.

\item In each selection query, if the query point is $q=(q_x,q_y,q_z)$,
then the sweep-plane is at $q_z$ and $(q_x,q_y)$ is one of the corners
in the 2-D shallow cutting of $P$ at $q_z$ after patching.
The query integer $k' \leq 10k$ will be such that for an $x$-selection query, the reported 
point $(q_x',q_y,q_z)$ will have $q'_x > q_x$. However, 
for the $y$-selection and the $z$-selection queries, the reported 
points $(q_x, q'_y,q_z)$ and $(q_x,q_y, q'_z)$, respectively, 
will have $q'_y< q_y$ and $q'_z < q_z$.
\end{enumerate}
\end{lemma}

The running time of the NR-construction can be summarized 
as follows: 
\begin{equation}
T_2(N) = \underbrace{T_{\mathcal{Q}}\left(\frac{N}{k}\right)}_{\text{queue}} + \underbrace{O\left(\frac{N}{k}\right)\cdot T_{sel}(N)}_{\text{bottleneck}}
\end{equation}
Similar to the second bottleneck in the
AT-construction, 
the bottleneck in the NR-construction
is the I/Os required to perform the dominance 
selection queries (with $O(sort(N))$ construction I/Os
for the data structures, we have
 $O\left(\frac{N}{k}\right)\cdot T_{sel}(N)= O\left(\frac{N}{k}\left(\frac{N}{B} \right)^{\gamma}\right)$ I/Os
 which is expensive).   

%\vspace{0.1 in}
%\noindent\textbf{Other constructions.} 
\subsection{Other constructions}
At ICALP'14, Afshani, Chan, and Tsakalidis~\cite{act14}
proposed an AT-construction variant (by sweeping upwards) for 
constructing shallow cuttings for 3-D dominance ranges in $O(N\log_2\log_2 N)$ 
time in the word-RAM model. At STOC'96, Vengroff and Vitter~\cite{vv96} 
employed a structure that is similar to
3-D shallow cuttings to answer 3-D
orthogonal range reporting queries. Unfortunately, the number of I/Os required
to construct their data structure was not considered in their work.

Construction of shallow cuttings for
halfspaces in 3-D has been another active area of research.
However, typically the techniques and tools used for these
problems (such as vertical decomposition, $\vare$-cuttings, 
planar graph separators) have been different from the ones used for 
3-D dominance ranges. For example,
at SoCG'15, Chan and Tsakalidis~\cite{ct15} 
presented an optimal deterministic
internal memory algorithm for the construction of shallow cuttings for
halfspaces in 3-D. 
%Their hierarchical construction requires the 
%cells in the cutting to be interior-disjoint throughout the 
%algorithm (unlike the overlapping cells constructed in our work). 
Their hierarchical construction relies on the fact that the cells
of the shallow cutting do not intersect (to be precise, 
they are interior-disjoint). 
While this is true in the case of shallow cuttings for 3-D
halfspaces, this property is not satisfied in our problem: cells 
of a $k$-level shallow cutting for 3-D dominance ranges are intersecting
boxes.
Also, in their hierarchy the levels
go down by a constant factor which is not ideal for an 
I/O-efficient algorithm, whereas in our construction
levels will go down at a faster rate.

%=================
%Optimal construction re-written by Rahul

\section{Our algorithm for constructing $k$-level 
shallow cuttings}\label{sec:our-algo}
We are now ready to prove Theorem~\ref{thm:shallow1}.
Our algorithm follows the general approach used 
in~\cite{at14,nr23} and combines it with several other 
ideas. First, we construct a \emph{hierarchy} of shallow cuttings 
such that the conflict list size of cells decreases exponentially. Second, 
somewhat counterintuitively, we  use data structures with high 
query cost and fast pre-processing for 3-D dominance 
selection queries. The final algorithm is a combination 
of six pieces.

%\vspace{0.1 in}
%\noindent\textbf{1st piece: Hierarchy of shallow cuttings.}
\subsection{1st piece: Hierarchy of shallow cuttings.} 
The main idea is to construct a hierarchy of shallow cuttings.
We want to construct a $k$-level shallow cutting of $P$ where 
$k\ge M/c$ and $c$ is a sufficiently large constant. 
Let $t_i=(N/B)^{{\delta}^i}$ and 
 $k_i=\max(Bt_i,k)$. We select the constant $\delta \in (0,1)$ in such way that $k_i \leq k_{i-1}/10$ (for any $i$ satisfying $k_i\ge M/c$). 
The algorithm works in stages (starting from stage~$0$). 
In the $i$-th stage, we 
construct the $k_i$-level shallow cutting.
For the starting case of $0$-th stage, since $k_0=N$,  
the $k_0$-level shallow cutting consists of one cell 
containing all points of $P$. 
The algorithm terminates when $k_i=k$.

\begin{lemma}\label{lem:number}
The number of shallow cuttings constructed is 
$r=O(\log_{M/B}(N/B))$.
\end{lemma}
\begin{proof}
Since $k_i=Bt_i \geq M/c$, we have $t_i \geq M/Bc$. Therefore,
$t_{i+1}=t_i^{\delta}=t_i\cdot \frac{1}{t_i^{1-\delta}}
\leq \frac{t_i}{(M/cB)^{1-\delta}}$. Unfolding the recursion, 
we obtain $t_{i+1} \leq \frac{t_0}{\left(\frac{M}{cB}\right)^{(i+1)(1-\delta)}}$. Setting $i\leftarrow \frac{c'}{1-\delta}\log_{M/cB}(N/B)-1$, 
for $c'\geq 2$,
we observe that $t_{i+1} \leq t_0/(N/B)^{c'}<1$.
This implies 
that $r=O(\frac{1}{1-\delta}\log_{M/B}(N/B))=O(\log_{M/B}(N/B)$.
\end{proof}
\begin{comment}
\begin{proof}
\tcr{polish it.}  Our algorithm constructs $k_i$-shallow cuttings for $1\le i\le r$ where  $r<f+O(1)$ and $f$ is the smallest integer such that $\frac{N}{B}^{\delta^f}\le M$.  Therefore $f$ is the smallest integer satisfying $M^{1/\delta^f}\ge (N/B)$ and \[f<\log_{1/\delta}\log_{M}\frac{N}{B}+1=O(\log_M \frac{N}{B})=O(\log_{M/B}\frac{N}{B})\] 
\end{proof}
\end{comment}

%\vspace{0.1 in}
%\noindent\textbf{2nd piece: Existence of parent cells.}
\subsection{2nd piece: Existence of parent cells.} We will start with some definitions.
For all $1\leq i\leq r$, we will say that the $k_{i-1}$-level shallow cutting is the \emph{parent} shallow cutting of the $k_i$-level shallow cutting. 
%A point $p$ 2-D dominates another point $q$
%if and only if $p$ has a higher coordinate value than $q$ 
%in both dimensions ($x$ and $y$).
When the sweep-plane reaches $z$, then 
a cell $C$ is {\em active} at $z$ if $cre(C) \leq z \leq exp(C)$.

In the construction of the $k_{i}$-level shallow cutting, 
consider an outer corner $(x_p, y_p)$ on the 2-D shallow cutting 
which got created when the sweep-plane is at position $z$. 
Then the {\em parent cell} of  point $p=(x_p,y_p,z)$ is the cell $C$ in the parent 
 shallow cutting which satisfies the following two properties: (a) the apex point of 
 $C$ dominates $p$, and  (b) the apex point of $C$ has the smallest $y$-coordinate among all cells that are active at $z$.

\begin{lemma}
 In the construction of the $k_i$-level shallow cutting (for $i\geq 1$), 
 consider an outer corner $(x_p, y_p)$ on the 2-D shallow cutting 
which got created when the sweep-plane is at position $z$. 
Then a parent cell always exists for the point $(x_p,y_p,z)$.
%consider the sweep-plane at any $z$. For each outer corner $(x_p, y_p)$
 %on the 2-D shallow cutting at  $z$, a parent cell always exists 
% for the point $(x_p,y_p,z)$.
\end{lemma}
\begin{proof}
The depth of the point $(x_p, y_p, z)$ will be less than or equal to $10k_i$ (an invariant maintained by the AT-construction), which is
at most $k_{i-1}$. In the $k_{i-1}$-level shallow cutting construction,
consider the 2-D shallow cutting when the sweep-plane reached $z$. 
This 2-D shallow cutting will have at least one outer corner, say
$p'$, which dominates $(x_p, y_p)$ (by third property of $k_{i-1}$-level
shallow cuttings). Let $C$ be the 3-D cell in the $k_{i-1}$-level 
shallow cutting corresponding to the outer corner $p'$. Since $C$ is 
active when the sweep-plane reached position $z$, it implies that 
the apex point of $C$ dominates $(x_p,y_p,z)$. This establishes
property~(a) of a parent cell, and hence, a parent cell exists.
\end{proof}

%\vspace{0.1 in}
%\noindent\textbf{3rd piece: Query expensive structures 
%on small-sized inputs.} 
\subsection{3rd piece: Query-expensive structures 
on small-sized inputs.} 
A data structure built on $N'$ points is query-expensive if the 
the number of I/Os required to answer a query is of the form
$O((N'/B)^{\gamma})$ or $O((N'/B)^{\gamma} + t/B)$.
Our next idea is to construct such
query-expensive data structures, but only on small-sized inputs. 
Specifically, for the $k_i$-level shallow cutting, 
we will ensure that the a 3-D dominance 
$x,y,z$-selection query requires only $O(k_i/B)$ I/Os, and 
since the number of such selection queries 
performed is bounded by $O(N/k_i)$ (point~2 in Lemma~\ref{lem:summary}),
overall the selection queries will require only $O(N/B)$ I/Os. Similarly, we will ensure that
generating the conflict list of a cell requires only 
$O(k_i/B)$ I/Os, and since there are $O(N/k_i)$ cells in the $k_i$-level
cutting, generating all the conflict lists will require only $O(N/B)$ I/Os.
The existence of parent cells and the following lemma will let us achieve the goal
of utilizing query-expensive data structures I/O-efficiently.

\begin{lemma}\label{lem:significance}
 In the construction of the $k_i$-level shallow cutting (for $i\geq 1$), 
consider the sweep-plane at any time $z$. If patching is performed at time $z$,
then let $p=(x_p, y_p)$ be a new outer corner created with parent cell $C$.
Then the 3-D dominance $x,y,z$-selection queries with query point 
$p=(x_p,y_p,z)$ will output a point that lies inside $C$. 
\end{lemma}
\begin{proof}
 This is trivial to see for a $y$-selection query since $p$ already 
 lies inside $C$ and hence, the point reported by the query will also lie
 inside $C$ (see Lemma~\ref{lem:summary}). Similarly, it holds 
 for a $z$-selection query as well. The non-trivial case is the $x$-selection 
query. 

Let $p''$ denote the point where the horizontal ray from $p$ hits the 2-D shallow cutting of the parent at time $z$. See Figure~\ref{fig:parentcor}. 
Let $d$ be the inner corner immediately below $p''$ and $c$ be the 
outer corner immediately above $p''$. Since the depth of 
$d$ is at least $k_{i-1}$, the depth of $p''$ will also be at least $k_{i-1}$. 
Consider the 3-D dominance $x$-selection query 
with query point $p$ and parameter $k' \leq 10k_i \leq k_{i-1}$.
This implies that the output, $p'$, will have an $x$-coordinate less than 
or equal to $p''$. Now the cell corresponding to the outer corner $c$
is the parent cell of $p$ and contains $p'$ as well.
\end{proof}

\begin{lemma}\label{lem:par-free}
Assume that the $k_{i-1}$-level shallow cutting of $P$ has been 
constructed in the $(i{-}1)$-th stage. 
Ignoring the I/Os required to find the parent cells, in the $i$-th stage, 
the $k_i$-level shallow cutting can be constructed in  $O\left(\frac{N}{k_i}\log_B\frac{N}{k_i}
+ \frac{N}{B} \log_{M/B}\frac{k_{i-1}}{B} \right)$ I/Os. 
\end{lemma}

\begin{proof}
For each cell in the $k_{i-1}$-level shallow cutting, 
 based on its conflict list we construct
data structures supporting  3-D dominance 
 reporting and 3-D dominance $x,y,z$-selection queries
(Theorem~\ref{thm:reporting}) that requires 
$O( (N/k_{i-1})\cdot \frac{k_{i-1}}{B}\log_{M/B}(
\frac{k_{i-1}}{B}))=O\left(\frac{N}{B} \log_{M/B}\frac{k_{i-1}}{B}\right)$
I/Os.  By point~2 in Lemma~\ref{lem:summary}, $O(N/k_i)$ $x,y,z$-selection queries
will be performed. For each selection query, if the cost of finding the parent 
cell of the query point is ignored, then by Lemma~\ref{lem:significance} 
the answer to the query can be 
obtained by querying the $x,y,z$-selection data structure built at 
the parent cell. The query I/Os of each selection query will then be 
$O((k_{i-1}/B)^{\gamma})$. Therefore, the number of I/Os performed
to answer the selection queries is 
$O\left(\frac{N}{k_i}\left(\frac{k_{i-1}}{B}\right)^{\gamma} \right)=
O\left(\frac{N}{Bt_i}\cdot t_{i-1}^{\gamma} \right)=
O\left(\frac{N}{Bt_i}\cdot t_{i-1}^{\delta} \right)=O\left(\frac{N}{B} \right)$, 
by choosing $\gamma < \delta$. 

By point~1 in Lemma~\ref{lem:summary}, the number of operations on the 
priority-queue are $O(N/k_i)$ and using a vanilla B-tree as a priority-queue
  $O\left(\frac{N}{k_i}\log_B\frac{N}{k_i}\right)$ I/Os are performed. 

Finally, for each cell $C$ in the $k_i$-level cutting, its conflict list is generated
by querying its parent cell's 3-D dominance reporting data structure.
The number of I/Os performed is $O\left(\left(\frac{k_{i-1}}{B} \right)^{\gamma} + \frac{k_i}{B} \right)=O(k_i/B)$. Since the number of cells in the $k_i$-level 
cutting is $O(N/k_i)$, the total number of I/Os required to generate the 
conflict lists is $O(N/B)$.
\end{proof}
 
\begin{figure}[tb]
  \centering
 \includegraphics[width=.4\textwidth]{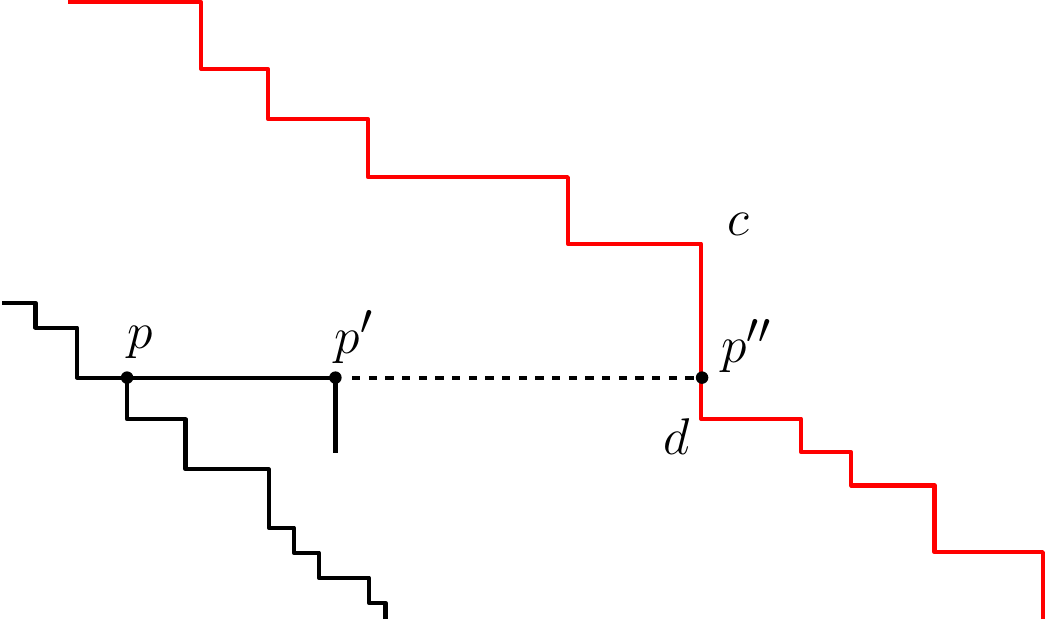}
  \caption{\label{fig:parentcor} An example of $x$-selection query in Lemma~\ref{lem:significance}. The point $p$ is an outer
  corner of a 2-D 
  shallow cutting of the $k_i$-level. The corner $c$ is the outer  corner of the parent staircase corresponding to the parent cell $C$ of $p$,
  since $c$ has the smallest $y$-coordinate among all apex points of the parent staircase dominating $p$.}
\end{figure}

%\vspace{0.1 in}
%\noindent\textbf{4th piece: Finding parent cells efficiently.} 
\subsection{4th piece: Finding parent cells efficiently.} 
By point~2 in Lemma~\ref{lem:summary},  
$O(N/k_i)$ $x,y,z$-selection queries are performed 
in the construction of a $k_i$-level
shallow cutting. To perform each selection query, 
we first need an efficient routine to find the parent cell.
Once the parent cell is found, then the selection query is 
performed on the data structure corresponding to the parent cell.

We will use partially persistent B-trees~\cite{bgo+96} to store the 
cells in the $k_{i-1}$-level shallow cutting.
Versions of the partially persistent data structure  are 
 parameterized by $z$-coordinates: a cell is inserted (resp.\ 
 deleted) at time $t$ if the corresponding outer corner is added to 
 the 2-D shallow cutting (resp.\ removed from the 2-D shallow cutting) 
 when the $z$-coordinate of the sweep-plane is equal to $t$.  This data 
 structure supports successor queries with respect to $y$-
 coordinates of its apex points: for any pair $(z_q,y_q)$ we can 
 find the cell with the smallest $y$-coordinate among all cells of 
 the parent shallow cutting that are  active at time $z_q$ and 
 whose $y$-coordinate is not smaller than $y_q$. 
 
 \begin{lemma}\label{lem:persistence}
 In the construction of $k_i$-level shallow cutting, 
the number of I/Os performed to compute the parent cells
is $\left(\frac{N}{k_{i}}\log_B\frac{N}{k_{i-1}} \right)$.
 \end{lemma}
\begin{proof}
The number of cells in the $k_{i-1}$-level shallow cutting 
is $O(N/k_{i-1})$ and using persistent B-trees, 
a successor query and an update query can be supported
in $O(\log_B (N/k_{i-1}))$ I/Os. By point~2 in 
Lemma~\ref{lem:summary}, the total number of $x,y,z$-selection 
queries will be $O(N/k_i)$, and hence, 
the number of I/Os performed to compute the parent cells
is $\left(\frac{N}{k_{i}}\log_B\frac{N}{k_{i-1}} \right)$.
\end{proof}

%\vspace{0.1 in}
%\noindent\textbf{5th piece: Obtaining log bounds with base $M/B$.}
\subsection{5th piece: Obtaining logarithmic bounds with the base $M/B$.} To obtain $O(sort(N))$ I/Os bound, we need the base in the logarithm
to be $M/B$. To circumvent the fact that finding parent cells and 
operations on the priority queue
require  $O(\log_B (N/k))$ query I/Os, we exploit the fact 
that $k$ is large.

\begin{lemma}\label{lem:b-is-fine}
For $k=\Omega(M)$, we have
$O\left(\frac{N}{k}\log_B\frac{N}{k}\right)
=O\left(\frac{N}{B}\log_{M/B}\frac{N}{B}\right)$.
\end{lemma}
\begin{proof}
First, we will establish that $B\log(M/B) \leq cM\log B$,
where $c$ is a sufficiently large constant. When $M \leq B^c$, 
we have $cM\log B = M\log B^c \geq 
M\log M \geq B\log (M/B)$. 
When $M > B^c$, we have 
$cM\log B \geq cB\cdot M^{1-1/c}\log B\geq cB\cdot M^{1-1/c} \geq cB\log M\geq cB\log(M/B)$. Using this fact, we have 
\[\frac{N}{k}\cdot\frac{\log (N/k)}{\log B} \leq \frac{N}{M}\cdot\frac{\log (N/B)}{\log B}
\leq O\left(\frac{N}{B}\frac{\log (N/B)}{\log (M/B)}\right).\qedhere\]
\end{proof}

In summary, the running time of our algorithm is
    \begin{equation}
T_3(N) = O\left(T_{\mathcal{Q}}\left(\frac{N}{k}\right) +
\sum_{i=1}^r \frac{N}{k_i}\cdot T_{sel}(k_{i-1}) + 
\frac{N}{k}\cdot \log_B\frac{N}{k}\right)
\end{equation}

\begin{lemma}  \label{lem:large-k}
For $k=\Omega(M)$, the $k$-level shallow cutting for 
3-D dominance ranges on $N$ points can be 
constructed in $\sort$ I/Os. 
\end{lemma}
\begin{proof}
By combining Lemma~\ref{lem:par-free} and Lemma~\ref{lem:persistence},
the number of I/Os required to construct a $k_i$-level shallow cutting is 
$O\left(\frac{N}{k_i}\log_B\frac{N}{k_i} + \frac{N}{B} \log_{M/B}\frac{k_{i-1}}{B} \right)$.
Therefore, 
\[ \sum_{i=1}^{r} O\left(\frac{N}{k_i}\log_B\frac{N}{k_i} + \frac{N}{B} \log_{M/B}\frac{k_{i-1}}{B}\right) =O\left(\frac{N}{k}\log_B\frac{N}{k} +  \frac{N}{B} \log_{M/B}
\frac{N}{B}\right)\]
By Lemma~\ref{lem:b-is-fine}, this can be bounded by 
$O\left(\frac{N}{B} \log_{M/B}\frac{N}{B}\right)$.
\end{proof}

%\vspace{0.1 in}
%\noindent\textbf{Final piece: Handling Small Values of $k$.} 
\subsection{Final piece: Handling Small Values of $k$.} We will need
a different procedure to handle $k=O(M)$. We will perform a two-level
cutting. First, we will construct an $(M/c)$-level shallow cutting of $P$
using Lemma~\ref{lem:large-k}. For each cell, say $C$, in the $(M/c)$-level, load its {\em entire} conflict list into the internal memory. This can be done I/O-efficiently since the size of the conflict list is $O(M)$. Next run an internal memory algorithm to compute the $k$-level shallow cutting of the conflict list of $C$. The final 
output is the union of the shallow cuttings obtained from each cell in the 
$(M/c)$-level.

\begin{lemma}
  \label{lem:small-k}
  Suppose that we are given an $(M/c)$-level shallow cutting of $P$ for some constant $c>1$.
  Then for $k< M/c$, the $k$-level shallow cutting can be constructed
  in $O(N/B)$ I/Os.
\end{lemma}
\begin{proof}
We will establish that the three properties of a $k$-level shallow cutting of $P$
are satisfied by the algorithm. The number of cells constructed in the final cutting is $O(N/M) \cdot O(M/k)=O(N/k)$ (first property). It is easy to verify that the number of points
of $P$ lying inside each cell in the final cutting is $O(k)$ (second property).
Consider a point $q$ that dominates less than or equal to $k$ points in $P$. 
 By the third property of $M/c$-level shallow cutting of $P$, there exists a cell 
 $C$ that contains $q$. The number of points in the conflict list of $C$ 
 dominated by $q$ will still be less than or equal to $k$. Therefore, in the $k$-level shallow 
 cutting of the conflict list of $C$, there will exist some cell that
 contains 
 $q$ (third property).

Loading the conflict lists of all the 
cells in the $(M/c)$-level shallow 
cutting requires $O(N/M)\cdot O(M/B)=O(N/B)$ I/Os. For each cell in the $(M/c)$-level 
shallow cutting, writing the $k$-level shallow cutting of its conflict-list 
into external memory 
requires $O(M/B)$ I/Os. Thus, over all the cells in the $(M/c)$-level shallow 
cutting, it will require $O(N/M)\cdot O(M/B)=O(N/B)$ I/Os.
\end{proof}

Combining Lemma~\ref{lem:large-k} and 
Lemma~\ref{lem:small-k} completes the proof of Theorem~\ref{thm:shallow1}.

%===============
%Applications
%===============

%=================
%First Application
%=================

\section{Application-1: Offline 3-D dominance approximate counting}
\label{sec:approx}
Recall that $P$ is a set of $N$ points in 3-D and $Q$ is a set of 
query 3-D dominance ranges. Consider a parameter $\vare$
with a fixed value in $(0,1)$. In the offline 3-D dominance
approximate counting
problem, the goal is to report a value in the range $[(1-\vare)|P \cap q|, |P\cap q|]$, for all $q\in Q$. We will prove 
Theorem~\ref{thm:main-approx-cnt} in this section.

\subsection{Data structure and query algorithm}
%\vspace{0.1 in}
%\noindent\textbf{Data structure.} 
The overall data structure 
is recursive and we start at the root with pointset $P$.
Construct a $(N^{2/3}B^{1/3})$-level cutting $\cC$ based on $P$ (Theorem~\ref{thm:shallow1}). 
%Based on the cells in the cutting, construct 
%an offline-find-any structure (Theorem~\ref{lem:find-any}). 
Also, construct 
an additive error counting structure based on 
$P$ (Theorem~\ref{thm:add-approx}). Finally, for each cell 
$C \in \cC$, recurse on $P_C$, where $P_C$ is the 
conflict list of $C$. The recursion stops when the size 
of pointset becomes less than or equal to $M/c$, where 
$c$ is a sufficiently large constant.

%\vspace{0.1 in}
%\noindent\textbf{Query algorithm.} 
The query algorithm will start from the 
root of the data structure.
A point in $Q$ is {\em deep} if no cell in $\cC$ contains it; otherwise, it is {\em shallow}. Using the 
offline-find-any procedure on $\cC$ and $Q$, we will classify each 
point in $Q$ as shallow or deep. The result for the
deep query points are obtained by querying the 
additive error counting structure. If a point $q\in Q$ is shallow, then it is {\em assigned} to one of the cells containing it. The shallow points are handled recursively. 
Specifically, for each cell $C\in \cC$, {\em recurse} on $Q_c$, where $Q_c \subseteq Q$ is the set of 
shallow query points assigned to $C$.

To handle the base case, where $|P| \leq M/c$, first 
the points in $P$ are loaded into the memory. Then the query
objects are streamed into the 
main memory one block at a time, and
the results for these queries are obtained without 
any additional I/Os and written
to the external memory.

\subsection{Analysis.}
%\vspace{0.1 in}
%\noindent\textbf{Analysis.} 
We will first bound the 
number of I/Os performed by the algorithm.
\begin{lemma}\label{lem:const-veb}
The number of I/Os required to construct the data structure
is $O\left(\frac{|P|}{B}\left(\log_{M/B}\frac{|P|}{B}\right)^{O(1)}\right)$.
\end{lemma}
\begin{proof}
 The $(N^{2/3}B^{1/3})$-level cutting (Theorem~\ref{thm:shallow1})
and the additive error counting structure (Theorem~\ref{thm:add-approx})
can be constructed in $O(sort(|P|)$ I/Os. 
Therefore, the total number of I/Os 
required to construct the data structure is:
\begin{align*}
T(N) = \begin{cases}
    c_1\left(\frac{N}{B}\right)^{1/3}\cdot T(N^{2/3}B^{1/3}) + \frac{c_2N}{B}\log_{M/B}\frac{N}{B}, \text{if $N\geq M/c$} \\
    O(N/B), \text{ otherwise}
\end{cases}
\end{align*}
where $c_1(N/B)^{1/3}$ is the number of cells in the cutting
and $c_2$ is the constant inside $O(sort(N))$.
We will establish that $T(N) \leq \frac{c_3N}{B}
\left(\log_{M/B}\frac{N}{B} \right)^{c_4}$, where 
$c_3=2c_2$ and $c_1\left(\frac{2}{3}\right)^{c_4} = \frac{1}{2}$.
As such, 
\begin{align*}
T(N) &\leq c_1\left(\frac{N}{B}\right)^{1/3}\cdot c_3\frac{N^{2/3}B^{1/3}}{B}
\left(\log_{M/B}\frac{N^{2/3}B^{1/3}}{B} \right)^{c_4} 
+\frac{c_2N}{B}\log_{M/B}\frac{N}{B}\\
&\leq c_1c_3\frac{N}{B}\cdot 
\left(\frac{2}{3}\right)^{c_4}\left(\log_{M/B}\frac{N}{B} \right)^{c_4} + 
\frac{c_2N}{B}\log_{M/B}\frac{N}{B}\\
&\leq \frac{c_2N}{B} \left(\log_{M/B}\frac{N}{B} \right)^{c_4} + 
\frac{c_2N}{B}\left(\log_{M/B}\frac{N}{B}\right)^{c_4}
=\frac{c_3N}{B}
\left(\log_{M/B}\frac{N}{B} \right)^{c_4}\qedhere
\end{align*}
\end{proof}

\begin{lemma}
The number of I/Os performed by the query algorithm 
is \\
$O_{\vare}\left(\frac{|P|}{B}\left(\log_{M/B}\frac{|P|}{B}\right)^{O(1)}+ 
\frac{|Q|}{B}\log_{M/B}\frac{N}{B}\right)$.
\end{lemma}
\begin{proof}
By Theorem~\ref{thm:add-approx} and Lemma~\ref{lem:find-any}, 
the number of I/Os performed by the query algorithm at the root is
 $O_{\vare}\left(sort(N) + sort(|Q|) +  \left( \frac{|\cC| + |Q|}{B}\right)\log_{M/B}\frac{|\cC|}{B}\right)$, which can be bounded by \\
$O_{\vare}\left(sort(N) + \frac{|Q|}{B}\log_{M/B}\frac{N}{B} \right)$.
The $O_{\vare}(sort(N))$ I/Os can be charged to the construction 
of the data structure, and hence, by the proof of 
Lemma~\ref{lem:const-veb},
it will be bounded by $O_{\vare}\left(\frac{|P|}{B}\left(\log_{M/B}\frac{|P|}{B}\right)^{O(1)}\right)$
in the overall query algorithm.

Therefore, we will focus on the $O_{\vare}\left(\frac{|Q|}{B}\log_{M/B}\frac{N}{B}\right)$ term.
As such, the amortized number of I/Os performed for a {\em single} query point is 
$O_{\vare}\left(\frac{1}{B}\log_{M/B}\frac{N}{B}\right)$.
Including recursion, the number of I/Os performed to answer a 
single query is: 
    
\[
q(N) \leq q(N^{2/3}B^{1/3}) + O_{\vare}\left(\frac{1}{B}\log_{M/B}\frac{N}{B}\right) =
O_{\vare}\left(\frac{1}{B}\log_{M/B}\frac{N}{B}\right), 
\] 
which holds because we end up with a geometrically decreasing series.
\end{proof}
%Now we will prove the correctness of the algorithm.
\begin{lemma}
For each $q\in Q$, the algorithm returns a $(1-\vare)$-approximation of $|P \cap q|$.
\end{lemma}
\begin{proof}
For any deep query point $q$ at the root, the value $n_q$ 
reported by the additive error counting structure will be such that 
$n_q \geq |P \cap q| -\vare N^{2/3}B^{1/3} \geq (1-\vare)|P\cap q|$,
since $|P\cap q| \geq N^{2/3}B^{1/3}$ for a deep point (third
property of shallow cuttings). The same argument applies
at each node in the data structure (w.r.t. deep queries
at that node).
\end{proof}

%%%%%%%%

%=================
%Second Application
%=================
\section{Application-2: Offline 3-D dominance reporting}
\label{sec:reporting}
Let $P$ be a set of $N$ points in 3-D and let $Q$ be a set of 
3-D dominance ranges. In the offline 3-D dominance reporting
problem, the goal is to report $P \cap q$, for all $q\in Q$.
We will prove Theorem~\ref{thm:main-reporting} in this section.

%\vspace{0.1 in}
%\noindent\textbf{Data structure.} 
\subsection{Data structure}
We will use 
two data structures as black-box to design our final 
data structure.
\begin{enumerate}
    \item For all $0\leq i\leq \ell$,
    we will construct $k_i$-level shallow cutting based on $P$ 
    (Theorem~\ref{thm:shallow1}) 
    and carefully choose $\ell$
    such that:
    \[  k_i=Bt_i, \quad t_i=(N/B)^{\delta^i} \text{ for some 
    $\delta \in (0,1)$} \text{ and }  
    k_{\ell}= \max\{M,\log_{M/B}(N/B)\}.\]

    \item For all $0\leq i \leq \ell$, based
    on the conflict list of each cell in 
    the $k_i$-level, construct a 3-D dominance 
    reporting structure (Theorem~\ref{thm:reporting}).

    %\item For all $0 \leq i \leq \ell$, based on 
    %the cells in the $k_i$-level shallow cutting, 
    %build an instance of the 
    %offline-find-any structure (Theorem~\ref{lem:find-any}). 
\end{enumerate}

%\vspace{0.1 in}
%\noindent\textbf{Query algorithm.}
\subsection{Query algorithm}
For each dominance range $q \in Q$, the goal is to
find the integer
$j$ such that the apex point of $q$ does not lie 
inside any cell in the $k_{j+1}$-level cutting but 
lies inside at least one cell in the $k_{j}$-level cutting.
Initially, let $Q_{\ell} \leftarrow Q$. 
The algorithm will work in iterations and in each iteration
we will process the queries in $Q$ in a batched manner 
as follows:\\

%\vspace{0.1 in}
\noindent
(1)  In the $i$-th 
iteration ($0\leq i\leq \ell$), 
partition the dominance ranges in $Q_{\ell-i}$ into a 
{\em shallow set}  and a {\em deep set}. 
 A dominance range in $Q_{\ell -i}$ is classified 
as shallow if the apex point of the dominance range 
lies inside some cell in the $(k_{\ell-i})$-level cutting;
otherwise, it is classified as deep.\\ 
(2) Run the offline-find-any procedure
on $Q_{\ell-i}$ and the cells in the 
$(k_{\ell-i})$-level cutting (Theorem~\ref{lem:find-any}).
This procedure partitions the 
set $Q_{\ell-i}$
into shallow and deep, and 
for each shallow dominance query 
$q \in Q_{\ell-i}$, reports a cell 
$C_q$ in the cutting that
contains $q$. \\
(3) Let $Q^s_{\ell-i} \subseteq Q_{\ell-i}$ 
be the set of shallow ranges. Invoke 
\textsc{Report}($P,Q^s_{\ell-i}$), that reports
$P \cap q$, for all $q\in Q^{s}_{\ell-i}$.
Assign $Q_{\ell-i-1}$ to be the deep set in 
$Q_{\ell-i}$ and continue to the 
next iteration. \\
%(3) For each shallow dominance query $q \in Q_{\ell-i}$, 
%invoke \textsc{Report}($P,q,C_q$), which reports 
%$P \cap q$. 
%Assign $Q_{\ell-i-1}$ to be the deep set in 
%$Q_{\ell-i}$ and continue to the next iteration. \\
(4) At the end of the $\ell$-th iteration, $Q_{-1}$ is 
the set of unanswered queries. 
For each $q \in Q_{-1}$, do a brute-force scan of $P$ 
to report $P\cap q$ (since $|P \cap q|=\Omega(N)$). \\

%\vspace{0.1 in}
\noindent
\textsc{Report}($P, Q^s_{\ell-i}$): \\
1. For all $q\in Q^s_{\ell-i}$, report $P \cap q$ by the querying the 
3-D dominance reporting structure built for $C_q$.\\
2. As an exception, we will handle the case of
$k_{\ell}=M$ and $i=0$ differently. We will process 
each cell in the $(k_{\ell})$-level one-by-one.
Let $C$ be the current cell being processed.
Load the entire conflict 
list of $C$ is into the main-memory. The queries 
in $Q^s_{\ell-i}$ assigned to cell $C$ are read one
block at a time into the main-memory. Using
any internal memory algorithm, the 3-D dominance
reporting query is answered for the queries 
in the block, and the output is
written back to external memory.

%\vspace{0.1 in}

%%%%%%%%%%%%%
%\noindent\textbf{Analysis.} 
\subsection{Analysis}
The easy part of the analysis is bounding
the number of I/Os required to construct the data structure.
\begin{lemma}
The number of I/Os required to construct the 
data-structure is $O(sort(|P|)$.
\end{lemma}
\begin{proof}
Using Theorem~\ref{thm:shallow1}, in $O(sort(|P|))$ I/Os all the $k_i$-level 
shallow cuttings can be constructed. The number of 
 I/Os required to construct the 3-D dominance reporting 
 structure at a cell in the $k_i$-level shallow cutting 
 is $(sort(k_i))$ I/Os. Therefore, the number of I/Os 
 required to construct the 3-D dominance reporting
 structures across all the $\ell+1$ levels will be
 \[ \sum_{i=0}^{\ell} O\left(\frac{N}{k_i}\cdot sort(k_i) \right)=
 \sum_{i=0}^{\ell} O\left(\frac{N}{k_i}\cdot 
 \frac{k_i}{B}\log_{M/B}(k_i/B)\right)=O(sort(N)). \qedhere\]
 \end{proof}
 
 The non-trivial part of the analysis is bounding the number of I/Os
 to perform $\ell+1$ iterations of offline-find-any procedure. The
 $j$-th iteration of offline-find-any procedure has two terms: 
 (a) the {\em first term} 
 is $O\left(\frac{|\cC_{\ell-j}|}{B}\cdot\log_{M/B}\frac{|\cC_{\ell-j}|}{B}\right)$, 
 and (b) the {\em second term} is 
 $O\left(\frac{|Q_{\ell-j}|}{B}\cdot\log_{M/B}\frac{|\cC_{\ell-j}|}{B}\right)$. We will 
 analyze both the terms separately.

\begin{lemma}\label{lem:first-term}
Adding up the first term over $\ell+1$ iterations is equal to $O(sort(|P|))$.
\end{lemma}
\begin{proof}
There are two cases to consider.
First, assume $k_{\ell}=\log_{M/B}(N/B)\geq M$. Then, we 
have 
\[\sum_{i=0}^{\ell}O\left(\frac{|P|}{Bk_i}\log_{M/B}\frac{|P|}{Bk_i}\right)
=O\left(\frac{|P|}{B}\log_{M/B}\frac{|P|}{B}\right)
\cdot \sum_{i=0}^{\ell}\frac{1}{k_i}=
O\left(\frac{|P|}{B}\log_{M/B}\frac{|P|}{B}\right),\]
since $\sum_{i=0}^{\ell}\frac{1}{k_i} 
\leq \frac{\ell}{k_{\ell}}
= \frac{O\left(\log\left(\frac{\log(N/B)}{\log(k_{\ell}/B)} \right)\right)}{\log_{M/B}(N/B)} 
=O\left( \frac{\log_{k_{\ell}/B}(N/B)}{\log_{M/B}(N/B)}\right)
=O\left(\frac{\log(M/B)}{\log(k_{\ell}/B)}\right)
=O(1)$.

Next, assume that $k_{\ell}=M \geq \log_{M/B}(N/B)$. 
Then, we have $\ell=\log \left(\frac{\log(N/B)}{\log(M/B)} \right)=\log_{M/B}(N/B)$, and hence,
\[\sum_{i=0}^{\ell}O\left(\frac{|P|}{Bk_i}\log_{M/B}\frac{|P|}{Bk_i}\right)= \left(\frac{\ell}{M}\right)\cdot O\left(\frac{|P|}{B} \log_{M/B}\frac{|P|}{B}\right)=O\left(\frac{|P|}{B}\log_{M/B}\frac{|P|}{B}\right).\qedhere \]
\end{proof}

Now we will analyze the second term.
Consider a query point $q$ that was deep in the $i$-th 
iteration, but became shallow in the $(i+1)$-th iteration.
Then, $\frac{|P\cap q|}{B}=\Omega\left(\frac{k_{\ell-i}}{B}\right)$.
The number of I/Os performed by the 
offline-find-any procedure in the $j$-th iteration, for 
any $j\leq i$, will be $O\left(\frac{|Q_{\ell-j}|}{B}\log_{M/B}\frac{|P|}{B} \right)$.
As such, {\em amortized cost} associated with $q$ 
for the offline-find-any
procedure is $\frac{1}{|Q_{\ell-j}|}\cdot O\left(\frac{|Q_{\ell-j}|}{B}\log_{M/B}\frac{|P|}{B} \right)=O\left(\frac{1}{B}\log_{M/B}\frac{|P|}{B} \right)$.
For $i$ iterations, the total amortized cost 
associated with $q$ will be 
$O\left(\frac{i}{B}\log_{M/B}\frac{|P|}{B} \right)$.
The following lemma establishes that the total amortized 
cost associated with $q$ for participating in  $i$ iterations 
of the offline-find-any 
procedure can be charged to the size of the output reported
for query $q$. 
%(the filtering search idea of Chazelle~\cite{c86}).

\begin{lemma}\label{lem:charging}
Consider a query point $q$ that is deep in the $i$-th iteration, but shallow 
in the $(i+1)$-th iteration. Then,  
$\frac{i}{B}\log_{M/B}(N/B) =O\left( \frac{k_{\ell-i}}{B}\right)
=O\left(\frac{|P\cap q|}{B}\right)$.
\end{lemma}
\begin{proof}
For all $0\leq j\leq \ell-2$, we first establish that
%For all $0\leq i \leq \ell-2$, we first establish that
\begin{equation}\label{eqn:diff}
k_j - k_{j+1} \geq k_{j+1} - k_{j+2}.
\end{equation}
Since $\frac{M}{B}=\Omega(1)$, we claim that
$t_{\ell}=\frac{k_{\ell}}{B}=\frac{1}{B}\max\{M,\log_{M/B}(N/B)\} \geq 
2^{1/1-\delta}$. This in turn implies
$t_j \geq t_{\ell} \geq 2^{1/1-\delta}$, then
$t_j^{1-\delta} \geq 2$ and finally
$t_j \geq 2t_j^{\delta}$. Therefore,
\[t_j + t_j^{\delta^2} \geq 2t_j^{\delta} \implies
t_j - t_{j+1} \geq t_{j+1} - t_{j+2} \implies 
k_j - k_{j+1} \geq k_{j+1} - k_{j+2}.\]
Via similar calculations it can be established that
\begin{equation}\label{eqn:base}
k_{\ell-1}-k_{\ell} \geq k_{\ell}.
\end{equation}

Finally, by combining Equations~\ref{eqn:diff} and \ref{eqn:base}, we observe that 
$k_{\ell-i} -k_{\ell} = \sum_{j=\ell-i}^{\ell-1}(k_j-k_{j+1})
=\Omega(ik_{\ell})$, which implies $k_{\ell-i}=\Omega(ik_{\ell})=\Omega\left(i\log_{M/B}(N/B) \right)$.
\end{proof}

\begin{lemma}\label{lem:sec-term}
Adding up the second term over $\ell+1$ iterations is equal to 
$O(K/B)$.
\end{lemma}
\begin{proof}
For each query $q\in Q$, let $q(i)$ be the number of iterations of 
offline-find-any procedures $q$ participates. Then, by the 
amortization argument, we have 
$\sum_{j=0}^{\ell}O\left(\frac{|Q_{\ell-j}|}{B}\log_{M/B}\frac{N}{B} \right)
=\sum_{q\in Q}O\left(\frac{q(i)}{B}\log_{M/B}\frac{N}{B}\right)$.
Via Lemma~\ref{lem:charging}, we have 
$\sum_{q\in Q}O\left(\frac{q(i)}{B}\log_{M/B}\frac{N}{B}\right)
=\sum_{q\in Q}O\left(\frac{|P\cap q|}{B} \right)=O(K/B)$.
\end{proof}

\begin{lemma}\label{lem:rep-qt}
The number of I/Os performed by the query algorithm is 
$O\left(sort(N)+ \frac{|Q|}{B}\log_{M/B}\frac{N}{B} + \frac{K}{B}\right)$.
\end{lemma}
\begin{proof}
The 
 first case we consider is $i=0$ and $k_{\ell}=M$.
Then the number of I/Os performed is dominated by 
the I/Os needed to (a) read the conflict list of all 
the cells in $k_{\ell}$-level, (b) read the queries 
in $Q_{\ell}$, and (c) write the output back to the 
external memory. In total, this requires only 
$O\left(\frac{|P| + |Q|}{B} + \frac{K}{B}\right)$ I/Os.

 For all $0\leq i \leq \ell$, recall that $Q_{\ell-i}^s \subseteq Q_{\ell-i}$ 
 is the shallow set in the $i$-th iteration.
 The second case we consider is $i=0$ and 
 $k_{\ell}=\log_{M/B}(N/B)$.
 For each $q\in Q_{\ell}^s$ , 
 querying the 3-D dominance reporting structure 
built on $P \cap C_q$ requires
\[O\left(|Q_{\ell}^s|\left(\frac{\log_{M/B}(N/B)}{B} \right)^{\gamma} + \sum_{q\in Q_{\ell}^s}\frac{|P\cap q|}{B} \right)
\leq O\left(\frac{|Q|}{B}\log_{M/B}\frac{N}{B} + \frac{K}{B}\right) \text{ I/Os}. \]

The final case we consider is $i\in [1,\ell]$. 
Consider any $i\in [1,\ell]$. Then for any 
$q\in Q_{\ell-i}^s$,
\begin{align*}
\left(\frac{k_{\ell-i}}{B} \right)^{\gamma}
+ \frac{|P\cap q|}{B} \leq t_{\ell-i}^{\gamma}
+ \frac{|P\cap q|}{B} \leq 
(t_{\ell-i+1})^{\gamma/\delta} + \frac{|P\cap q|}{B} \leq
t_{\ell-i+1} + \frac{|P\cap q|}{B} 
\leq 2\cdot\frac{|P\cap q|}{B},
\end{align*}
where we use the fact that $\delta \geq \gamma$
and $|P\cap q|\geq k_{\ell-i+1}$.
Therefore, for all $1\leq i\leq \ell$ and 
for all $q\in Q_{\ell-i}^s$, 
querying the 3-D dominance reporting structure 
built on $P \cap C_q$ requires
\[\sum_{i=1}^{\ell}\sum_{q\in Q_{\ell-i}^s} 
O\left( \left(\frac{k_{\ell-i}}{B} \right)^{\gamma}+\frac{|P\cap q|}{B}\right)
=\sum_{i=1}^{\ell}\sum_{q\in Q_{\ell-i}^s} O\left(\frac{|P\cap q|}{B}\right)=O(K/B) \text{ I/Os}.\]

Finally, from Lemma~\ref{lem:first-term} and 
Lemma~\ref{lem:sec-term}, the number of 
I/Os performed by all the iterations of the 
offline-find-any query procedure is bounded by 
$O\left( sort(N) + \frac{K}{B}\right)$.
\end{proof}

%=====================
%Supporting Structures
%=====================
\section{A common framework for supporting structures}\label{sec:support}
In this section we will construct data structures that
are needed in the algorithms for constructing the $k$-level 
shallow cutting in Section~\ref{sec:our-algo} 
and the two applications 
(in Section~\ref{sec:approx} and \ref{sec:reporting}). 
We will present a common framework for constructing these
data structures, and in the process 
prove Theorem~\ref{thm:reporting} and Theorem~\ref{thm:add-approx}.

\begin{figure}[h]
    \centering
    \includegraphics{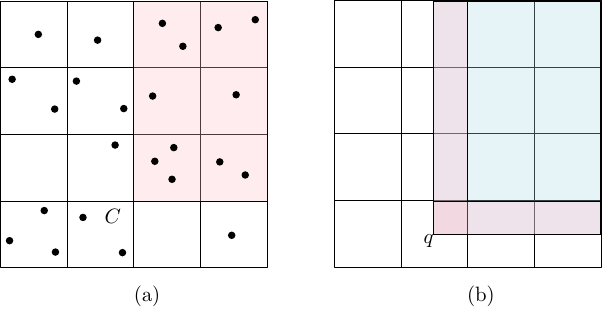}
    \caption{For simplicity, the root node is shown in 2-D.
    (a) An example with $24$ points and $\alpha=5$. If the apex point of a query dominance range lies 
    in cell $C$, then the cells shaded pink lie completely inside
    the query region. Therefore, for cell $C$ we have $n(C)=10$. (b) The blue portion of the query region is handled at the root itself. We will recurse on the pink portion.}
    \label{fig:root-node}
\end{figure}

\subsection{Proof of Theorem~\ref{thm:reporting}}
3-D dominance $x,y,z$-selection queries can be trivially
reduced to $O(\log_2N)$ 3-D dominance counting queries 
(via binary search). Therefore, we will focus our discussion 
on the reporting and the counting queries.
Our final data structure will be a tree structure.
We will first describe the root node of the structure. 

\vspace{0.1 in}
\noindent\textbf{The root node.}
%\paragraph{The root node.} 
Let $P$ be a set of $N$ points in 3-D. Let $\cX=\{X_1, X_2,\ldots,X_{\alpha}\}$ be 
a set of $\alpha$ planes such that 
(a) they are perpendicular to the $x$-axis, 
(b) for all $i< j$, plane $X_i$ has a smaller $x$-coordinate
than $X_j$ , and 
(c) the number of points of $P$ lying between two consecutive planes is $N/(\alpha-1)$.
The value of $\alpha$ will be determined later.
Analogously define  a  set $\cY=\{Y_1,\ldots, Y_{\alpha}\}$ 
consisting of $\alpha$ planes w.r.t. the $y$-axis 
and a set $\cZ=\{Z_1,\ldots,Z_{\alpha}\}$ consisting 
of $\alpha$ planes w.r.t. the $z$-axis. 
For all $1\leq i \leq \alpha$, let $X_i^{+}$ and $X_i^{-}$ be the locus of points in 3-D 
that have a higher and lower $x$-coordinate, respectively, than the $x$-coordinate 
of the plane $X_i$. Analogously, for all $1\leq i\leq \alpha$, define 
$Y_i^+, Y_i^{-}, Z_i^{+}, \text{ and } Z_i^{-}$. Then, for all 
$1 \leq i,j,k \leq \alpha{-}1$, a {\em cell} $(i,j,k)$ 
is the cuboid $(X_i^+, X_{i+1}^-) \times (Y_j^+, Y_{j+1}^-) \times 
(Z_k^+, Z_{k+1}^-)$. Sets 
$\cX,\cY, \text{ and } \cZ$ can be constructed 
trivially in $sort(N)$ I/Os by sorting $P$ in the 
respective dimension. See Figure~\ref{fig:root-node}(a).

Define {\em conflict list} of a cell to be the points of $P$ lying 
inside the cell. The next task is to compute the conflict list of 
each cell. 
\begin{itemize}
\item For all $1\leq i \leq \alpha{-}1$, compute the conflict list, say $P_i$, of 
$(X_i^+, X_{i+1}^-) \times (-\infty, +\infty) \times 
(-\infty, +\infty)$.  This can done in $O(N/B)$ I/Os by scanning the 
sorted list of $P$ along the $x$-axis.

\item Fix an $i$ in the range  $[1,\alpha{-}1]$. 
Sort the points in $P_i$ 
along the $y$-axis. For all $1\leq j \leq \alpha{-}1$, 
compute the conflict list, say $P_{ij}$, of $(X_i^+, X_{i+1}^-) \times (Y_j^+, Y_{j+1}^-) \times 
(-\infty, +\infty)$. Perform this step for all values of $i$ in the range 
$[1,\alpha{-}1]$. The number of I/Os performed will be 
$\sum_{i=1}^{\alpha} O\left(\frac{|P_i|}{B}\log_{M/B}\frac{|P_i|}{B} + \alpha \right)
=O\left(\frac{N}{B}\log_{M/B}\frac{N}{B} + \alpha^2\right)$.

\item Finally, fix two integers $i$ and $j$ in the range $[1,\alpha{-}1]$. 
Sort the points in $P_{ij}$ along the $z$-axis. 
For all $1\leq k \leq \alpha{-}1$, 
compute the conflict list, say $P_{ij}$, of $(X_i^+, X_{i+1}^-) \times (Y_j^+, Y_{j+1}^-) \times 
(Z_k^+, Z_{k+1}^-)$. Perform this step for all values of $i$ and $j$ in the range 
$[1,\alpha{-}1]$. The number of I/Os performed will be 
$\sum_{i=1}^{\alpha} \sum_{j=1}^\alpha O\left(\frac{|P_{ij}|}{B}\log_{M/B}\frac{|P_{ij}|}{B} + \alpha \right)
=O\left(\frac{N}{B}\log_{M/B}\frac{N}{B} + \alpha^3\right)$.
\end{itemize}

\vspace{0.1 in}
\noindent\textbf{Recursive step.}
%\paragraph{Recursive step.} 
At the root node, the sets 
$\cX, \cY, \text{ and } \cZ$, and the conflict list (along
with its size) of each cell is stored. 
Next, recursively construct the data structure based on the following 
$3(\alpha-1)$ subsets of $P$: for all $1\leq i \leq \alpha{-}1$,
(a)  the conflict list of 
$(X_i^+, X_{i+1}^-) \times (-\infty, +\infty) \times 
(-\infty, +\infty)$, (b)  the conflict list of 
$(-\infty, +\infty) \times (Y_j^+, Y_{j+1}^-) \times (-\infty, +\infty)$, 
and (c) the conflict list of 
$(-\infty, +\infty) \times  (-\infty, +\infty) \times (Z_k^+, Z_{k+1}^-)$. The recursion stops when the number of points fall below $B\alpha^3$.

\begin{lemma}\label{lem:rep-cnt}
By setting $\alpha=(N/B)^{\gamma/c}$ (for a sufficiently large constant $c$),  
the number of I/Os needed to build the 
data structure is $O_{\gamma}(sort(N))$.
\end{lemma}
\begin{proof}
Let $h$ be the height of the data structure. The number of points associated 
with a node at depth $\ell$ will be $N/\alpha^{\ell}$. Therefore,
\[N/\alpha^{h} = B\alpha^3 \implies \frac{N}{B}=\alpha^{h+3} \implies h=\frac{c}{\gamma} -3. \]
The number of I/Os needed to build the root node is 
$O(sort(N) + \alpha^3)=O(sort(N))$. Therefore, 
the amortized I/Os needed per point in $P$ will be $O(sort(N)/N)$.
Since a point belongs to $3^h=O_{\gamma}(1)$ nodes in the tree, 
the amortized I/Os needed per point in the entire tree will be 
$O_{\gamma}(sort(N)/N)$. Hence, the overall number of I/Os
is $O_{\gamma}(sort(N))$.
\end{proof}

\begin{comment}
\begin{proof}
We will recurse only twice and stop. Then the number of points
at a node at the leaf will be $\frac{N}{(N/B)^{2/3}}=N^{1/3}B^{2/3} \leq \frac{\vare}{c}\sqrt{NB}$ (by assuming $N/B=\Omega_{\vare}(1)$, because otherwise the query can be answered trivially via
a linear-scan). The number of I/Os needed to build the 
root node is $O(sort(N) + \alpha^3)=O(sort(N))$. The remaining proof
follows from the
same amortization argument as in the proof of Lemma~\ref{lem:rep-cnt}.
\end{proof}
\end{comment}

\vspace{0.1 in}
\noindent\textbf{Query algorithm.} 
%\paragraph{Query algorithm.} 
Assume that the query region of the 
form $q=[q_x,\infty) \times [q_y,\infty) \times [q_z,\infty)$. 
At the root node of the data structure, let $C_q$ be the cells 
that lie completely inside $q$. For the reporting query, 
we report the conflict list of all the cells in $C_q$.
%For the counting query and the additive error query, 
Let $n(C_q)$ be the total size of
the conflict list of the cells in $C_q$ (this value can 
be precomputed and stored at each cell). 

Let $(i,j,k)$ be the 
cell containing the apex point $(q_x,q_y,q_z)$. Define $R_i=(X_i^+, X_{i+1}^-) \times (-\infty, +\infty) \times 
(-\infty, +\infty)$,  
$R_j=(-\infty, +\infty) \times (Y_j^+, Y_{j+1}^-) \times (-\infty, +\infty)$, 
and  
$R_k=(-\infty, +\infty) \times  (-\infty, +\infty) \times (Z_k^+, Z_{k+1}^-)$.
Then we recursively query the children corresponding to $R_i, R_j, \text{ and } R_k$
with query regions $q \cap R_1, (q\setminus R_1) \cap R_2,  \text{ and } (q\setminus R_1 \setminus R_2) \cap R_3$, respectively. (This is done to avoid reporting duplicates or counting a point multiple times.) See Figure~\ref{fig:root-node}(b). When a leaf node is visited, 
then a linear scan 
of the pointset is performed.
The output of the counting query
will be the sum of (a) $n(C_q)$'s over all the nodes visited, 
and (b) the number of points in the leaf nodes that lie
inside $q$.

\begin{lemma}
The number of I/Os to perform  a 
3-D dominance reporting query and a 
3-D dominance counting query on $N$ points 
is $O((N/B)^{\gamma} + t/B)$ and $O((N/B)^{\gamma})$, respectively, 
where $t$ is the number of points reported.
\end{lemma}
\begin{proof}
Ignoring the I/Os needed to report the points,
the I/Os performed at any non-leaf or a leaf node is $O(\alpha^3)$.
Since the number of nodes visited is $3^h=O_{\gamma}(1)$, the 
total  query cost is $O_{\gamma}(\alpha^3)=O((N/B)^{\gamma})$.
\end{proof}

\subsection{Proof of Theorem~\ref{thm:add-approx}}
For the offline additive error counting query, data
structure will simply be the 
root node as constructed above. Importantly, we 
do not recurse. By setting $\alpha=\frac{3}{\vare}\left(\frac{N}{B}\right)^{1/3}$,
the number of I/Os needed to build the 
data structure is $O(sort(N) + \alpha^3)=O_{\vare}(sort(N))$.

Let $Q$ be the set of 3-D dominance queries. 
The approach used to compute the conflict list of each cell 
(w.r.t. $P$) can be adapted to compute in $O_{\vare}(sort(N) + sort(|Q|)$
I/Os the conflict list of each cell (w.r.t. $Q$, i.e., locating
the cell containing the apex point of each query in $Q$).
Then for each query $q\in Q$, we will report $n(C_q)$ as the estimate.
It is easy to verify that $|P \cap q| - n_q \leq 
\frac{3N}{\alpha}=\vare N^{2/3}B^{1/3}$.

\vspace{0.1 in}
\noindent\textbf{Remark.} 
%\paragraph{Remark.} 
This data structure 
can be extended to handle smaller values of error, such as $\sqrt{NB}$, if we recurse for a constant number of 
levels.

%%%%%%%%%%%%%%%%%%%%
%%CONCLUSION
%%%%%%%%%%%%%%%%%%%%

\section{Conclusion and future work}\label{app:future}
We hope that these results will lead to further 
    work on 3-D orthogonal range searching and the related 
    problems in the I/O-model. We finish by mentioning some 
    future directions:
    \begin{itemize}
    	\item An immediate direction is to design an optimal
	algorithm for the offline 3-D dominance approximate counting problem.
	\item Studying 3-D dominance reporting and  3-D dominance approximate 	counting in the data structure setting where the queries arrive one after the other. Currently, efficient constructions of these data structures are not known.
        \item For 3-D orthogonal range reporting, 
        the existing state-of-the art
    data structure in the I/O-model~\cite{aal09}  is sophisticated
    and requires construction of ``colored'' version of shallow cuttings.
    Can their data structure be 
        constructed in  $O\left(\frac{|P|}{B}\left(\log_{M/B}\frac{|P|}
    {B}\right)^{O(1)} \right)$ I/Os? 
    \item For the offline version of 3-D orthogonal range 
        reporting, what are the right I/O-bounds to aim for? 
        Unlike 2-D, it might not be possible to match the sorting I/O-bound.
        
        \item Computing skyline or maximal points of $P$ in 3-D in the I/O model 
        has received attention~\cite{st11a, hst+13}. Skyline points and the $k$-level 
        shallow cuttings for dominance ranges seem related in their structures.
        Is there a common framework to design algorithms for both 
        the problems? The existing I/O-algorithms for maximal points
        use a divide-and-conquer approach.
        \item Shallow cuttings for 3-D halfspace ranges is another
        important problem in computational geometry. As mentioned
        in the Introduction, Chan and Tsakalidis~\cite{ct15} designed
        an optimal-cost internal memory algorithm for the construction
        of shallow cuttings for halfspaces in 3-D. Some 
        aspects of their algorithm are well suited to adapt to the I/O-model (such as the hierarchical construction where there
        is flexibility to choose the factor by which the levels fall).
        However, it looks
        non-trivial to I/O-efficiently adapt some of the other steps.

\end{itemize}

\paragraph{Acknowledgments.} The authors thank the reviewers whose detailed feedback for the conference version 
of the paper helped in improving the presentation of the paper.

%%
%% Bibliography
%%

%% Please use bibtex, 

%\bibliography{ref}

%Used for arxiv
\bibliographystyle{plain}
\bibliography{ref}

@string{cg="{Computational Geometry}"}

@string{esa="{ESA}"}

@string{focs = "{FOCS}"}

@string{icalp="{ICALP}"}

@string{icdt="{ICDT}"}

@string{pods="{PODS}"}

@string{sigmod="{SIGMOD}"}

@string{socg="{S}o{CG}"}

@string{soda="{SODA}"}

@string{stoc="{STOC}"}

@string{vldbj="{VLDB J.}"}

@string{esa="Proceedings of European Symposium on Algorithms ({ESA})"}

@string{focs="Proceedings of Annual {IEEE} Symposium on Foundations of Computer Science ({FOCS})"}

@string{icalp="Proceedings of International Colloquium on Automata, Languages and Programming ({ICALP})"}

@string{icdt="Proceedings of International Conference on Database Theory ({ICDT})"}

@string{pods="Proceedings of ACM Symposium on Principles of Database Systems ({PODS})"}

@string{sigmod="Proceedings of ACM Management of Data ({SIGMOD})"}

@string{socg="Proceedings of Symposium on Computational Geometry ({S}o{CG})"}

@string{soda="Proceedings of the Annual {ACM-SIAM} Symposium on Discrete Algorithms ({SODA})"}

@string{stoc="Proceedings of {ACM} Symposium on Theory of Computing ({STOC})"}

@string{vldbj="The {VLDB} Journal"}

@inproceedings{aaefv98,
  author       = {Pankaj K. Agarwal and
                  Lars Arge and
                  Jeff Erickson and
                  Paolo Giulio Franciosa and
                  Jeffrey Scott Vitter},
  title        = {Efficient Searching with Linear Constraints},
  booktitle    = pods,
  year         = {1998},
  pages        = {169--178}
}

@inproceedings{aal09,
  author    = {Peyman Afshani and
               Lars Arge and
               Kasper Dalgaard Larsen},
  title     = {Orthogonal Range Reporting in Three and Higher Dimensions},
  booktitle = focs,
  year      = {2009},
  pages     = {149-158}
}

@article{ahz10,
  author    = {Peyman Afshani and
               Chris H. Hamilton and
               Norbert Zeh},
  title     = {A general approach for cache-oblivious range reporting and approximate
               range counting},
  journal   = cg,
  volume    = {43},
  number    = {8},
  pages     = {700--712},
  year      = {2010}
}

@article{bgo+96,
  author    = {Bruno Becker and
               Stephan Gschwind and
               Thomas Ohler and
               Bernhard Seeger and
               Peter Widmayer},
  title     = {An Asymptotically Optimal Multiversion {B}-Tree},
  journal   = vldbj,
  volume    = {5},
  number    = {4},
  year      = {1996},
  pages     = {264-275}
}

@inproceedings{hst+13,
  author    = {Xiaocheng Hu and
               Cheng Sheng and
               Yufei Tao and
               Yi Yang and
               Shuigeng Zhou},
  title     = {Output-sensitive Skyline Algorithms in External Memory},
  booktitle = soda,
  year      = {2013},
  pages     = {887--900}
}

@inproceedings{ks99,
  author    = {Kothuri Venkata Ravi Kanth and
               Ambuj K. Singh},
  title     = {Optimal Dynamic Range Searching in Non-replicating Index
               Structures},
  booktitle = icdt,
  year      = {1999},
  pages     = {257-276}
}

@inproceedings{p08,
  author    = {Mihai Patrascu},
  title     = {Succincter},
  booktitle = focs,
  year      = {2008},
  pages     = {305-313}
}

@inproceedings{p08b,
  author    = {Peyman Afshani},
  title     = {On Dominance Reporting in {3D}},
  booktitle = esa,
  year      = {2008},
  pages     = {41-51}
}

@inproceedings{r81,
  author    = {John T. Robinson},
  title     = {The {K-D-B}-Tree: A Search Structure For Large Multidimensional
               Dynamic Indexes},
  booktitle = sigmod,
  year      = {1981},
  pages     = {10-18}
}

@inproceedings{r15,
  author    = {Saladi Rahul},
  title     = {Improved Bounds for Orthogonal Point Enclosure Query and Point Location
               in Orthogonal Subdivisions in $\mathbb{R}^3$},
  booktitle = soda,
  pages     = {200--211},
  year      = {2015}
}

@inproceedings{rt15,
  author    = {Saladi Rahul and
               Yufei Tao},
  title     = {On Top-k Range Reporting in 2D Space},
  booktitle = pods,
  pages     = {265--275},
  year      = {2015}
}

@inproceedings{st11a,
  author    = {Cheng Sheng and
               Yufei Tao},
  title     = {Finding Skylines in External Memory},
  booktitle = pods,
  year      = {2011},
  pages			= {107-116}
}

@inproceedings{t12b,
  author    = {Yufei Tao},
  title     = {Indexability of 2D range search revisited: constant redundancy
               and weak indivisibility},
  booktitle = pods,
  year      = {2012},
  pages     = {131-142}
}

@article{v06,
  author    = {Jeffrey Scott Vitter},
  title     = {Algorithms and data structures for external memory},
  journal   = {{F}oundation and {T}rends in {T}heoretical {C}omputer {S}cience},
  volume    = {2},
  number    = {4},
  pages     = {305-474},
  year      = {2006}
}

@inproceedings{rt16,
  author       = {Saladi Rahul and
                  Yufei Tao},
  title        = {Efficient Top-k Indexing via General Reductions},
  booktitle    = pods,
  pages        = {277--288},
  year         = {2016}
}

@inproceedings{nr23,
  author       = {Yakov Nekrich and
                  Saladi Rahul},
  title        = {4D Range Reporting in the Pointer Machine Model in Almost-Optimal
                  Time},
  booktitle    = soda,
  pages        = {1862--1876},
  year         = {2023}
}

@inproceedings{bfgmmt14,
  author       = {Michael A. Bender and
                  Martin Farach{-}Colton and
                  Mayank Goswami and
                  Dzejla Medjedovic and
                  Pablo Montes and
                  Meng{-}Tsung Tsai},
  title        = {The Batched Predecessor Problem in External Memory},
  booktitle    = esa,
  pages        = {112--124},
  year         = {2014}
}

@inproceedings{act14,
  author       = {Peyman Afshani and
                  Timothy M. Chan and
                  Konstantinos Tsakalidis},
  title        = {Deterministic Rectangle Enclosure and Offline Dominance Reporting
                  on the {RAM}},
  booktitle    = icalp,
  pages        = {77--88},
  year         = {2014}
}

@inproceedings{vv96,
  author       = {Darren Erik Vengroff and
                  Jeffrey Scott Vitter},
  title        = {Efficient 3-D Range Searching in External Memory},
  booktitle    = stoc,
   pages        = {192--201},
   year         = {1996}
}

@inproceedings{at14,
  author       = {Peyman Afshani and
                  Konstantinos Tsakalidis},
  title        = {Optimal Deterministic Shallow Cuttings for 3D Dominance Ranges},
  booktitle    = soda,
  pages        = {1389--1398},
  year         = {2014}
}

@inproceedings{cnrt18,
  author       = {Timothy M. Chan and
                  Yakov Nekrich and
                  Saladi Rahul and
                  Konstantinos Tsakalidis},
  title        = {Orthogonal Point Location and Rectangle Stabbing Queries in 3-d},
  booktitle    = icalp,
   pages        = {31:1--31:14},
  year         = {2018}
}

@inproceedings{ct15,
  author       = {Timothy M. Chan and
                  Konstantinos Tsakalidis},
  title        = {Optimal Deterministic Algorithms for 2-d and 3-d Shallow Cuttings},
  booktitle    = {International Symposium on Computational Geometry (SoCG)},
  pages        = {719--732},
  year = {2015}
}

@inproceedings{gi98,
  author       = {Roberto Grossi and
                  Giuseppe F. Italiano},
  title        = {Efficient cross-trees for external memory},
  booktitle    = {External Memory Algorithms},
  volume       = {50},
  pages        = {87--106},
  publisher    = {{DIMACS/AMS}},
  year         = {1998}
}
%Used for arxiv

\newpage 

\appendix
%%%%%APPENDIX%%%%%%%%

\end{document}